\documentclass{llncs}

\usepackage{makeidx}  

\usepackage{sgame}
\usepackage{enumerate}
\usepackage{color}
\usepackage{pstricks,pst-node}

\usepackage{url}
\usepackage{algorithmic}
\usepackage{amsmath}
\usepackage{amsfonts}
\usepackage{amssymb}

\usepackage{graphics}
\usepackage{latexsym}
\usepackage{epsf}
\usepackage{epsfig}

\newcommand{\eat}[1]{}
\newcommand{\squishlist}{
 \begin{list}{$\bullet$}
  { \setlength{\itemsep}{0pt}
     \setlength{\parsep}{3pt}
     \setlength{\topsep}{3pt}
     \setlength{\partopsep}{0pt}
     \setlength{\leftmargin}{1.5em}
     \setlength{\labelwidth}{1em}
     \setlength{\labelsep}{0.5em} } }

\newcommand{\squishlisttwo}{
 \begin{list}{$\bullet$}
  { \setlength{\itemsep}{0pt}
     \setlength{\parsep}{0pt}
    \setlength{\topsep}{0pt}
    \setlength{\partopsep}{0pt}
    \setlength{\leftmargin}{2em}
    \setlength{\labelwidth}{1.5em}
    \setlength{\labelsep}{0.5em} } }

\newcommand{\squishend}{
  \end{list}  }

\newcommand{\bfe}[1]{\begin{bfseries}\emph{#1}\end{bfseries}\index{#1}}
\newcommand{\oldbfe}[1]{\begin{bfseries}\emph{#1}\end{bfseries}}

\newcommand{\ES}{\mbox{$\emptyset$}}

\newcommand{\myra}{\mbox{$\:\rightarrow\:$}}

\newcommand{\La}{\mbox{$\:\Leftarrow\:$}}
\newcommand{\Ra}{\mbox{$\:\Rightarrow\:$}}
\newcommand{\tra}{\mbox{$\:\rightarrow^*\:$}}

\newcommand{\A}{\mbox{$\ \wedge\ $}}

\newcommand{\LL}{\mbox{$\ldots$}}

\newcommand{\C}[1]{\mbox{$\{{#1}\}$}}           

\newcommand{\NI}{\noindent}

\newcommand{\II}{\vspace{2 mm}}

\newcommand{\szkew}[1]{\relax \setbox0=\hbox{\kern -24pt $\displaystyle#1$\kern 0pt }%
\box0}
{\catcode`\@=11 \global\let\ifjusthvtest@=\iffalse}

\newcounter{oldmycaption}

\def\smallromani{\renewcommand{\theenumi}{\roman{enumi}}
\renewcommand{\labelenumi}{(\theenumi)}}

\newcounter{symbol}
\newcommand{\indexsyma}[1]%
{\stepcounter{symbol}\index{zzz1 \thesymbol @\protect#1}}
\newcommand{\indexsymb}[1]%
{\stepcounter{symbol}\index{zzz2 \thesymbol @\protect#1}}
\newcommand{\indexsymc}[1]%
{\stepcounter{symbol}\index{zzz3 \thesymbol @\protect#1}}
\newcommand{\indexsymd}[1]%
{\stepcounter{symbol}\index{zzz4 \thesymbol @\protect#1}}
\newcommand{\indexsyme}[1]%
{\stepcounter{symbol}\index{zzz5 \thesymbol @\protect#1}}

\input{epsf}

\begin{document}
\title{Diffusion in Social Networks with Competing Products}
\titlerunning{Diffusion in Social Networks with Competing Products}  
%
\author{Krzysztof R. Apt\inst{1,2} \and Evangelos Markakis\inst{3}
}
\authorrunning{K. R. Apt and E. Markakis} 
%
\tocauthor{Krzysztof R. Apt, Evangelos Markakis}
\institute{CWI, Amsterdam, the Netherlands,\\
\and
University of Amsterdam,\\
\email{apt@cwi.nl},\\
\and
Athens University of Economics and Business,\\
Dept. of Informatics, Athens, Greece\\
\email{markakis@gmail.com}}

\maketitle              


\begin{abstract}

We introduce a new threshold model of social networks, in which the
nodes influenced by their neighbours can adopt one out of several
alternatives.  We characterize the graphs
for which adoption of a product by the
whole network is possible (respectively necessary) and the ones
for which a unique outcome is guaranteed.
These characterizations directly yield
polynomial time algorithms that allow us to determine whether a
given social network satisfies one of the above properties.

We also study algorithmic questions for networks without unique
outcomes. We show that the problem of computing the
minimum possible spread of a product is NP-hard
to approximate with an approximation ratio better than $\Omega(n)$, in contrast to the maximum spread, which is efficiently computable.
We then move on to questions regarding the behavior of a node with respect to adopting some (resp. a given) product. We show
that the problem of determining whether a given node has to adopt some
(resp.~a given) product in all final networks is co-NP-complete.

\end{abstract}

% ...

\section{Introduction}

\subsection{Background}

Social networks have become a huge interdisciplinary research area with
important links to sociology, economics, epidemiology, computer
science, and mathematics.  
A flurry of numerous articles and recent books \cite{Jac08,EK10}
shows the growing relevance of
this field as it deals with such diverse topics as epidemics, spread of
certain patterns of social behaviour, effects of advertising, and
emergence of `bubbles' in financial markets.

A large part of research on social networks focusses on the problem of
\emph{diffusion}, that is the spread of a certain event or information
over the network, e.g., becoming infected or adopting a
given product. In the remainder of the paper, we will use as a running example the
adoption of a product, which is being marketed over a social network.

Two prevalent models have been considered for capturing diffusion: the
{\em threshold models} introduced in~\cite{Gra78} and~\cite{Sch78} and
the {\em independent cascade models} studied in~\cite{GLM01}. In
threshold models, which is the focus of our work, each node $i$ has a threshold $\theta(i) \in (0,1]$
and it decides to adopt a product when the total weight of incoming
edges from nodes that have already adopted a product reaches or exceeds $\theta(i)$.  In a special case a node decides to adopt a product if
at least the fraction $\theta(i)$ of its neighbours has done so.
In cascade models, each node that adopts a product can activate
each of his neighbours with a certain probability and each node has
only one chance of activating a neighbour.

Most of research has focussed on the situation in which the players
face the choice of adopting a specific product or not. In this setting, the algorithmic
problem of choosing an initial set of nodes so as to maximize the
adoption of a given product and certain variants of this were studied
initially in~\cite{KKT03}
and in several publications
that followed, e.g., \cite{Che09,MR07}.

When studying social networks from the point of view of adopting new
products that come to the market, it is natural to lift the restriction of a single product.
One natural example is when users choose among competing programs from providers of mobile
telephones. Then, because of lower subscription costs, each owner of
a mobile telephone naturally prefers to choose the same
provider that his friends choose. 
In such situations, the outcome of the
adoption process does not need to be unique. Indeed, individuals with
a low 'threshold' can adopt any product a small group of their friends
adopts. As a result this leads to different considerations than
before.


In the presence of multiple products, diffusion has been investigated
recently for cascade models in~\cite{BKS07,CNWZ07,KOW08}, where new
approximation algorithms and hardness results have been proposed.
For threshold models, an extension to two products has
been recently proposed in~\cite{BFO10}, where the authors examine whether the
algorithmic approach of~\cite{KKT03} can be extended.
Algorithms and hardness of approximation results
are provided for certain variants of the diffusion process.

Game theoretic aspects have also been considered in the case of two
products.  In particular, the behavior of best response dynamics in
infinite graphs is studied in~\cite{Mor00}, when each node has to
choose between two different products. An extension of this model is
studied in~\cite{IKMW07} with a focus on notions of compatibility and
bilinguality, i.e., having the option to adopt both products at an
extra cost so as to be compatible with all your neighbours.

\subsection{Contributions}

We study a new model of a social network in which
nodes (agents) can choose out of \emph{several} alternatives
\emph{and} in which various outcomes of the adoption process are
possible. Our model combines a number of features present in
various models of networks.

It is a threshold model and we assume that the threshold of a node is
a fixed number as in~\cite{Che09} (and unlike \cite{KKT03,BFO10}, where they are random variables).
This is in contrast to Hebb's
model of learning in networks of neurons, the focus of
which is on learning, leading to strengthening of the connections
(here thresholds). In our context threshold should be viewed as a fixed `resistance
level' of a node to adopt a product. In contrast to the SIR model, see,
e.g., \cite{Jac08}, in which a node can be in only two states,
in our model each node can choose out of several
states (products). We also allow that not all nodes have exactly the same set of products to choose from, e.g. due to geographic or income restrictions some products may be available only to a subset of the nodes. If a node changes its state from the initial one, the new
state (that corresponds to the adopted product) is final, as is the case with most of the related literature.

Our work consists of two parts. In the first part (Sections~\ref{sec:reachable}, \ref{sec:unavoidable}, \ref{sec:unique}) we study three basic problems concerning this model. In particular, we find necessary and sufficient conditions for determining whether

\squishlisttwo
\item a specific product will possibly be adopted by all nodes.

\item a specific product will necessarily be adopted by all nodes.

\item the adoption process of the products will yield a unique outcome.
\squishend

For each of these questions, we obtain a characterization with respect to the structure of the underlying graph.

In the second part (Section~\ref{sec:adoption-analysis}), we focus on networks that do not possess a unique outcome and investigate the complexity of various problems concerning the adoption process.
We start with estimating the minimum and maximum number of nodes that may adopt a given product.
Then we move on to questions regarding the behavior of a given node in terms of adopting a given product or some product from its list.
We resolve the complexity of all these problems. As we show, some of these problems are efficiently solvable, whereas the remaining ones are either co-NP-complete or have strong inapproximability properties.

\section{Preliminaries}

Assume a fixed weighted directed graph $G = (V, E)$ (with no parallel
edges and no self-loops), with $n = |V|$ and $w_{ij}\in [0, 1]$ being
the weight of edge $(i, j)$. 
Given a node $i$ of $G$ we denote by
$N(i)$ the set of nodes from which there is an incoming edge to $i$.
We call each $j \in N(i)$ a \oldbfe{neighbour} of $i$ in $G$.
We assume that for each node $i$ such that $N(i) \neq \ES$, $\sum_{j
\in N(i)} w_{ji} \leq 1$.
Further, we have a
\oldbfe{threshold function} $\theta$ that assigns to each node $i\in
V$ a fixed value $\theta(i) \in (0, 1]$.  Finally, we fix a finite set $P$
of alternatives to which we shall refer as \oldbfe{products}.

By a \oldbfe{social network} we mean a tuple $(G,P,p,\theta)$, where
$p$ is a function that assigns to each node of $G$ a non-empty subset
of $P$.
The idea is that each node $i$ is offered a non-empty set
$p(i)\subseteq P$ of products from which it can make its choice.  If
$p(i)$ is a singleton, say $p(i) = \{t\}$, the node adopted the
product $t$. Otherwise it can adopt a product if the total weight of
incoming edges from neighbours that have already adopted it is at
least equal to the threshold $\theta(i)$.
To formalize the questions we want to address, we need to introduce a number of
notions.  Since $G,P$ and $\theta$ are fixed, we often identify each
social network with the function $p$.

Consider a binary relation $\myra$ on social networks.  Denote
by $\tra$ the reflexive, transitive closure of $\myra$.  We call a
reduction sequence $p \tra p'$
\oldbfe{maximal} if for no $p''$ we have $p' \myra
p''$. In that case we will say that $p'$ is a \oldbfe{final} network, given the initial network $p$.


\begin{definition}
Assume an initial social network $p$ and  a network
$p'$. We say that
\squishlist
\item  $p'$ is \oldbfe{reachable} (from $p$) if $p \tra p'$,

\item $p'$ is \oldbfe{unavoidable} (from $p$)
if for all maximal sequences of reductions
$p \tra p''$ we have $p' = p''$,

\item $p$ has a \oldbfe{unique outcome}
if some social network is unavoidable from $p$.
\squishend
\end{definition}


From now on we specialize the relation $\myra$. Given a social
network $p$, and a product $t\in p(i)$ for some node $i$ with $N(i) \neq \ES$,
we use the abbreviation $A(t,i)$ (for `adoption condition
of product $t$ by node $i$') for
\[
\sum_{j \in N(i)\mid p(j) = \{t\} } w_{ji} \geq \theta(i)
\]
When $N(i) = \ES$, we stipulate that $A(t,i)$ holds for every $t\in p(i)$.

\begin{definition}
\mbox{}
\squishlisttwo
\item
We write $p_{1} \myra p_{2}$ if $p_2 \neq p_1$ and
for all nodes $i$, if $p_2(i) \neq p_1(i)$, then $|p_1(i)|  \geq 2$ and
for some $t \in p_1(i)$
\[
p_{2}(i) = \{t\} \mbox{ and $A(t,i)$ holds in $p_1$}.
\]

\item We say that node $i$ in a social network $p$
\begin{itemize}
\item
\oldbfe{adopted product $t$} if $p(i) = \{t\}$,

\item \oldbfe{can adopt product $t$} if $t \in p(i)$, $|p(i)| \geq 2$, and $A(t,i)$ holds in $p$.
\end{itemize}
\squishend
\end{definition}

In particular, a node $i$ with no neighbours and more than one product in $p(i)$
can adopt any product
that is a possible choice for it. 
Note that each modification of the function $p$ results in assigning to a node $i$
a singleton set. Thus, if $p_{1} \tra p_{2}$, then for all nodes
$i$ either $p_{2}(i) = p_{1}(i)$ or $p_{2}(i)$ is a singleton set.

One of the questions we are interested is whether a product $t$ can spread to the whole network.
We will denote this final network by $[t]$, where $[t]$ denotes the constant function $p$ such
that $p(i) = \{t\}$ for all nodes $i$.
Furthermore, given a social network $(G,P,p,\theta)$ and a product $t \in P$ we denote
by $G_{p,t}$ the weighted directed graph obtained from $G$ by removing from it
all edges to nodes $i$ with $p(i) = \{t\}$. That is, in $G_{p,t}$ for all such nodes $i$
we have $N(i) = \ES$ and for all other nodes the set of neighbours in $G_{p,t}$
and $G$ is the same.

If each weight  w$_{j,i}$ in the considered  graph
equals $\frac{1}{|N(i)|}$, then we call the corresponding social network \bfe{equitable}.
Hence in equitable social networks
the adoption condition, $A(t,i)$, holds if at least a fraction $\theta(i)$ of the neighbours of $i$ adopted in $p$ product $t$.

\begin{example}
\label{ex:nets}
As an example for illustrating the definitions, consider the equitable social networks in Figure
\ref{fig:soc}, where $P = \{t_1, t_2\}$ and where we mention next to
each node the set of
products available to it.

\begin{figure}[htbp]
\begin{center}
\ \setlength{\epsfxsize}{4.5cm}
\epsfbox{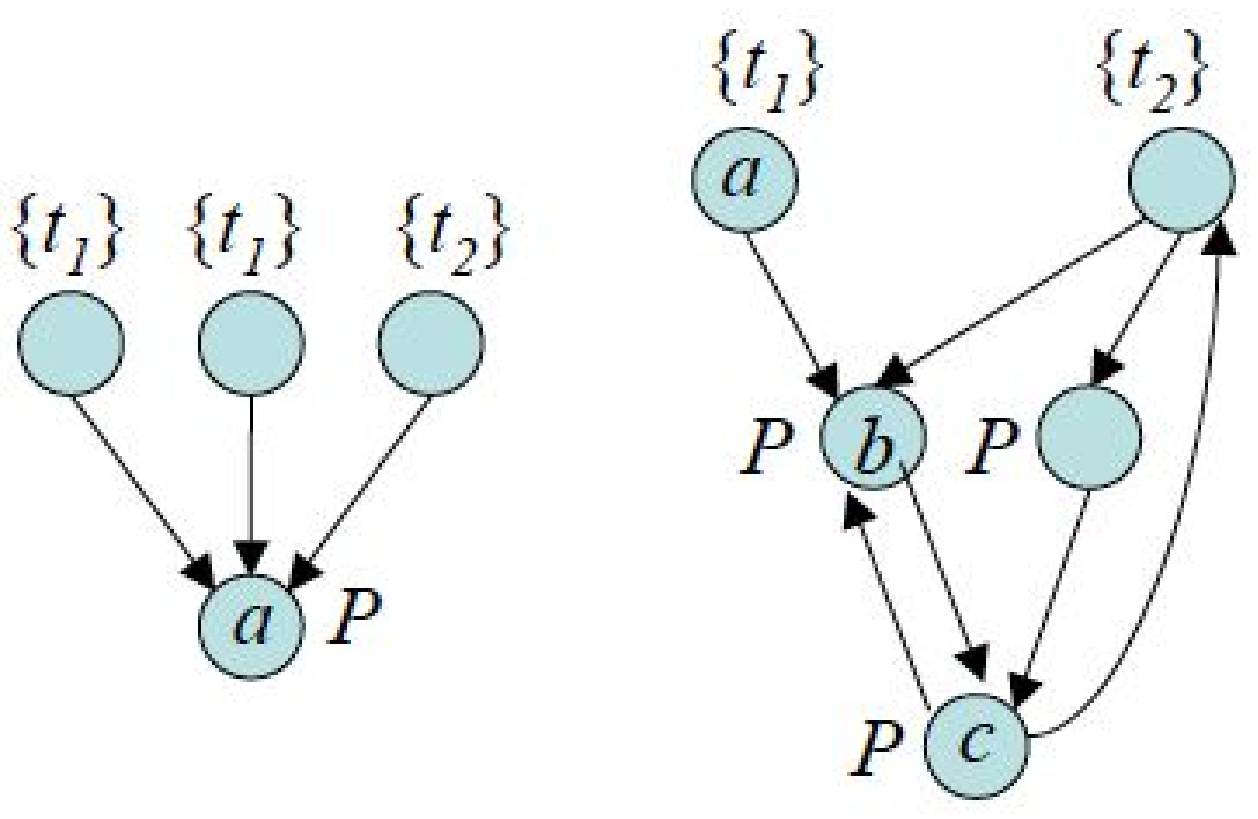}
\end{center}
\caption{Two examples of social networks} \label{fig:soc}
\end{figure}

In the first social network, if $\theta(a) \leq \frac{1}{3}$, then the
network in which node $a$ adopts product $t_1$ is reachable, and so is
the case for product $t_2$.  If $\frac{1}{3} < \theta(a) \leq
\frac{2}{3}$, then only the network in which node $a$ adopts product
$t_1$ is reachable.  Further, if $\theta(a) > \frac{2}{3}$, then none
of the above two networks is reachable.  Finally, the initial network
has a unique outcome iff $\frac{1}{3} < \theta(a)$.


For the second social network the following more elaborate case
distinction lists the possible values of $p$ in the final reachable
networks.


\[
\begin{array}{ll}
\theta(b) \leq \frac{1}{3}  \A \theta(c) \leq \frac{1}{2}& : (p(b) = \{t_1\} \vee p(b) = \{t_2\}) \A (p(c) = \{t_1\} \vee p(c) = \{t_2\}) \\[1mm]
\theta(b) \leq \frac{1}{3}  \A \theta(c) > \frac{1}{2} & : (p(b) = \{t_1\} \A p(c) = P) \: \vee (p(b) = p(c) = \{t_2\}) \\[1mm]
\frac{1}{3} < \theta(b) \leq \frac{2}{3}  \A \theta(c) \leq \frac{1}{2} & : p(b) = p(c) = \{t_2\} \\[1mm]
\frac{1}{3} < \theta(b) \A \theta(c) > \frac{1}{2} & : p(b) = p(c) = P \\[1mm]
\frac{2}{3} < \theta(b) \A \theta(c) \leq \frac{1}{2} & : p(b) = P \A p(c) = \{t_2\}
\end{array}
\]

In particular, when $\frac{1}{3} < \theta(b) \leq \frac{2}{3}$ and
$\theta(c) \leq \frac{1}{2}$, node $b$ adopts product $t_2$
only \emph{after} node $c$ adopts it.
\end{example}

\section{Reachable outcomes}
\label{sec:reachable}

We start with providing necessary and sufficient conditions for a product to be reachable by all nodes.
This is achieved by a structural characterization of graphs
that allow products to spread to the whole graph, given the threshold function $\theta$.
In particular, we shall need the following notion.

\begin{definition}
\label{def:q}
Given a threshold function $\theta$ we call a weighted directed graph
\oldbfe{$\theta$-well-structured} if for some function $level$ that maps
nodes to natural numbers, we have that for all nodes $i$ such that $N(i) \neq \ES$
\begin{equation}
\label{equ:theta}
\sum_{j \in N(i) \mid level(j) < level(i)} w_{ji} \geq \theta(i).
\end{equation}
\end{definition}

In other words, a weighted directed graph is $\theta$-well-structured if levels
can be assigned to its nodes in such a way that for each node $i$
such that $N(i) \neq \ES$, the
sum of the weights of the incoming edges from lower levels is at least
$\theta(i)$.  We will often refer to the function $level$ as a
\bfe{certificate} for the graph being $\theta$-well-structured. Note
that there can be many certificates for a given graph.
Note also that $\theta$-well structured graphs can have cycles. For instance, it is
easy to check that the second social network in Figure \ref{fig:soc} is
$\theta$-well structured when $\theta(i) \leq \frac{1}{3}$ for every node $i$.



We have the following characterization.

\begin{theorem} \label{thm:res1}
Assume a social network $(G,P,p,\theta)$ and a product $top \in P$.
A social network $(G,P,[top],\theta)$ is reachable from $(G,P,p,\theta)$ iff
\\[-6mm]
\begin{itemize}
\item for all $i$, $top \in p(i)$,
\item $G_{p,top}$ is $\theta$-well-structured.
\end{itemize}
\end{theorem}
\begin{proof}
\NI
$(\Ra)$
If for some node $i$ we have $top \not\in p(i)$,
then $i$ cannot adopt product $top$ and
$[top]$ is not reachable.

To establish the second condition consider a reduction
sequence
\[
p_{1} \myra p_{2} \myra \LL \myra p_{m}
\]
starting in $p$ and such that $p_{m} = [top]$.

Assign now to each node $i$ the minimal $k$ such that $p_{k+1}(i) =
\C{top}$.
We claim that this definition of the $level$ function shows that $G_{p, top}$ is
$\theta$-well-structured.
Consider a node $i$.
\II

\NI
\emph{Case 1.} $level(i) = 0$.

Then $p(i) = \C{top}$, so by the definition of $G_{p, top}$ we have $N(i) = \ES$ in $G_{p, top}$.
Hence we do not need to argue about these nodes since we only need to ensure condition (\ref{equ:theta}) for nodes with $N(i) \neq \ES$.
\II

\NI
\emph{Case 2.} $level(i) > 0$.

Suppose that $N(i) \neq \ES$
and that $level(i) = k$. By the
definition of the reduction $\myra$
the adoption condition $A(top,i)$ holds in $p_k$, i.e.,
\[
\sum_{j \in N(i)\mid p_k(j) = \{top\} } w_{ji} \geq  \theta(i).
\]


But for each $j \in N(i)$ such that $p_k(j) = \{top\}$ we have by
definition $level(j) < level(i)$.  So (\ref{equ:theta}) holds.
\II

\NI
$(\La)$
Consider a certificate function $level$ showing that $G_{p, top}$ is $\theta$-well-structured.
Without loss of generality we can assume that the nodes in  $G_{p, top}$ such that $N(i) = \ES$
are exactly the nodes of level $0$.
We construct by induction on the level $m$ a reduction sequence $p \tra p''$,
such that for all nodes $i$ we have $top \in p''(i)$ and for all nodes $i$
of level $\leq m$ we have $p''(i) = \C{top}$.

Consider level $0$.  By definition of $G_{p, top}$, a node $i$ is of level $0$ iff it has no
neighbours in $G$ or $p(i) = \C{top}$.  In the former case, by the
first condition, $top \in p(i)$. So $p \tra p''$, where the function
$p''$ is defined by
\[
        p''(i) :=
        \left\{
        \begin{array}{l@{\extracolsep{3mm}}l}
        \C{top}   & \mathrm{if}\  level(i) = 0 \\
        p(i)      & \mathrm{otherwise}
        \end{array}
        \right.
\]
This establishes the induction basis.

Suppose the claim holds for some level $m$. So we have $p \tra
p'$, where for all nodes $i$ we have $top \in p'(i)$ and for all
nodes $i$ of level $\leq m$ we have $p'(i) = \C{top}$.

Consider the nodes of level $m+1$.
For each such node $i$ we have  $top \in p'(i)$, $N(i) \neq \ES$ and
\[
\sum_{j \in N(i) \mid level(j) < level(i)} w_{ji} \geq \theta(i).
\]

By the definition of $G_{p, top}$
the sets of neighbours of $i$ in $G$ and $G_{p, top}$ are the same.
By the induction hypothesis
for all nodes $j$ such that $level(j) < level(i)$ we have $p'(j) = \{top\}$.

So either node $i$ adopted product $top$ in $p'$ or
can adopt product $top$ in $p'$.
Hence $p' \tra p''$, where
the function
$p''$ is defined by
\[
        p''(i) :=
        \left\{
        \begin{array}{l@{\extracolsep{3mm}}l}
        \C{top}   & \mathrm{if}\  level(i) = m+1 \\
        p'(i)      & \mathrm{otherwise}
        \end{array}
        \right.
\]
Consequently $p \tra p''$, which establishes the induction step.
We conclude $p \tra [top]$.
\end{proof}

Next we show that testing if a graph is $\theta$-well-structured can be
efficiently solved.

\begin{theorem}
\label{thm:theta}
Given a weighted directed graph $G$ and a threshold function $\theta$, we can decide whether
$G$ is $\theta$-well-structured in time $O(n^2)$.
\end{theorem}
\begin{proof}
(Sketch) We claim that the following simple algorithm achieves this:

\squishlisttwo
\item Given a weighted directed graph $G$, first assign level $0$ to
  all nodes with $N(i) = \emptyset$. If no such node exists, output
  that the graph is not $\theta$-well-structured.

\item Inductively, at step $i$, assign level $i$ to each node for
  which condition (\ref{equ:theta}) from Definition~\ref{def:q} is
  satisfied when considering only its neighbours that have been
  assigned levels $0,\LL,i-1$.

\item If by iterating this all nodes are assigned a level, then output
  that the graph is $\theta$-well-structured. Otherwise,
output that $G$ is not $\theta$-well-structured.
\squishend

The above algorithm can be implemented in time $O(n^2 + |E|) =
O(n^2)$, by using the adjacency list representation. 
To prove correctness, note that if the input
graph is not $\theta$-well-structured, then the algorithm will output
No, as otherwise, at termination it would have constructed a level function for a
non-$\theta$-well-structured graph. For the reverse,
suppose a graph $G$ is $\theta$-well-structured.
The idea of the proof is to use a certificate function, in which all nodes are assigned the minimum
possible level. We then prove by induction that this is precisely the level assignment produced by the algorithm and hence it outputs Yes. Due to lack of space, we omit the proof.

\end{proof}



Finally, we end this section by observing that
determining whether a network $[top]$ is reachable can also be solved efficiently.

\begin{theorem} \label{thm:res1-alg}
Assume a social network $(G,P,p,\theta)$ and a product $top \in P$.
There is an algorithm running in time  $O(n^2)$ that determines whether
the social network $(G,P,[top],\theta)$ is reachable.
\end{theorem}

\section{Unavoidable outcomes}
\label{sec:unavoidable}

Next, we focus on the notion of unavoidable outcomes.
We establish the following characterization.
\begin{theorem} \label{thm:res2}
Assume a social network $(G,P,p,\theta)$ and a product $top \in P$.
A social network $(G,P,[top],\theta)$ is unavoidable iff
\begin{itemize}
\item for all $i$, if $N(i) = \ES$, then $p(i) = \{top\}$,

\item for all $i$, $top \in p(i)$,

\item $G_{p,top}$ is $\theta$-well-structured.
\end{itemize}
\end{theorem}

To prove this, we need first a few lemmas, the proofs of which we omit from this version.

\begin{lemma} \label{lem:t}
Suppose that $p \tra p'$ and for some node
$i$ we have $p'(i) = \{t\}$.
Then for some node $j$ such that $N(j) = \ES$ or $p(j)$ is a singleton,
we have $t \in p(j)$.
\end{lemma}
Intuitively, this means that each product eventually adopted can also
be initially adopted (by a possibly different node).


\begin{lemma} \label{lem:top}
Assume a social network $(G,P,p,\theta)$ and a product $top \in P$.
Suppose that
\begin{itemize}
\item for all $i$, if $N(i) = \ES$ or $p(i)$ is a singleton, then $p(i) = \{top\}$.
\end{itemize}
Then a unique outcome of $(G,P,p,\theta)$ exists.
\end{lemma}
Intuitively, this means that if initially only one product can be adopted, then
a unique outcome of the social network exists.

\noindent {\bf Proof of Theorem~\ref{thm:res2}:} (Sketch)
By Theorem \ref{thm:res1} and Lemma \ref{lem:top}.



In analogy to Theorem~\ref{thm:res1-alg}, we also have the following simple fact.
\begin{theorem} \label{thm:res2-alg}
Assume a social network $(G,P,p,\theta)$ and a product $top \in P$.
There is an algorithm, running in time $O(n^2)$, that determines whether
the social network $(G,P,[top],\theta)$ is unavoidable.
\end{theorem}



\section{Unique outcomes}
\label{sec:unique}

We now consider the question of when does a network admit a unique outcome.
To answer this, we introduce the following definitions.

\begin{definition}
Given social networks $p, p'$ based on the same graph we say that
\begin{itemize}
\item
node $i$  \bfe{can switch in $p'$ given $p$} if $i$ adopted in $p'$ a product $t$ and for some $t' \neq t$
\[
\mbox{$t' \in p(i) \A$ $A(t',i)$ holds in $p'$,}
\]

\item $p'$ is \bfe{ambivalent given $p$} if it contains a node that
  either can adopt more than one product or can switch in $p'$ given $p$,

\item the reduction $p \myra p'$ is \bfe{fast} if for each node $i$, if
  $i$ can adopt a product in $p$ then $i$ adopted a product in $p'$. Intuitively, $p \myra p'$ is then a `maximal' one-step reduction of $p$.
\end{itemize}
\end{definition}

\begin{definition}
By the \bfe{contraction sequence} of a social network we mean the unique
reduction sequence $p \tra p'$ such that
\begin{itemize}
\item each of its reduction steps is fast,

\item either $p \tra p'$ is maximal or $p'$ is the first network in the sequence $p \tra p'$
that is ambivalent given $p$.
\end{itemize}
\end{definition}

We now formulate a characterization of social networks that admit
a unique outcome. We omit the proof.

\begin{theorem} \label{thm:res3}
A social network admits a unique outcome iff its contraction sequence ends in a
non-ambivalent social network.
\end{theorem}

\begin{corollary}
\label{cor:unique}
Assume a social network $(G,P,p,\theta)$ such that
\\[-6mm]
\begin{itemize}
\item for all nodes $i$ we have $\theta(i) > \frac{1}{2}$,
\item for all $i$, if $N(i) = \ES$, then $p(i)$ is a singleton.
\end{itemize}
Then $(G, P, p, \theta)$  admits a unique outcome.
\end{corollary}


The above corollary can be strengthened by assuming that the network is
such that if $\theta(i) \leq \frac{1}{2}$ then $|N(i)| < 2$ or $|p(i)|
= 1$. The reason is that the nodes for which $|N(i)| < 2$ or $|p(i)| =
1$ cannot introduce an ambivalence.

When for some node $i$, $\theta(i) \leq \frac{1}{2}$ holds and neither
$|N(i)| < 2$ nor $|p(i)| = 1$, the equitable
social network still may admit a unique outcome but it does not have
to.  For instance the second social network in
Figure \ref{fig:soc} admits a unique outcome for the last three alternatives (explained in Example~\ref{ex:nets}),
while for the first two is
does not.



\eat{
\begin{figure}
\hrule height0.8pt
\vspace{1pt}
\begin{algorithmic}[1]

\begin{scriptsize}
\STATE Produce the representation with a list of outgoing edges for each node;
\FOR{$i \in V$}
\STATE set $p(i)$ to be the initial list of products available to node $i$
\ENDFOR
\FOR{$j\in V, t\in p(j)$}
\STATE $S_{j,t} := 0$ ;// counts total weight to $j$ from nodes that adopted $t$;
\ENDFOR
\IF{$\exists i\in V$ with $N(i) = \emptyset$ and $|p(i)|\geq 2$}
\RETURN "No unique outcome";
\ENDIF
\STATE $L := \{ i\in V: |p(i)| = 1\}$ ;// initialize $L$ to a list of nodes that already have adopted a product;
\STATE \textbf{if} {$L = \emptyset$} \textbf{return} "Unique outcome" \textbf{endif};
\WHILE{$L\neq \emptyset$}
\STATE $R := \emptyset$;
\FOR{$i\in L$ and $j$ such that $(i, j)\in E$}
\STATE \textbf{if} $i$ has adopted $t$ and $t\in p(j)$ \textbf{then} $S_{j,t} := S_{j,t} + w_{ij} $\textbf{ end if};
\STATE $R := R\cup \{j\}$; // nodes we need to check for ambivalence
\ENDFOR

\FOR{$j\in R$}
\STATE Compute $|\{t: S_{j, t} \geq \theta(j)\}|$ ; //even for nodes that have already adopted a product
\STATE \textbf{if} {$|\{t: S_{j, t} \geq \theta(j)\}| \geq 2$} \textbf{return} "No unique outcome" \textbf{endif};
\IF{$|\{t: S_{j, t} \geq \theta(j)\}| = 1$ and $j$ has not yet adopted $t$}
\STATE node $j$ adopts product $t$;
\ELSE
\STATE $R := R\setminus \{j\}$ ;  // $j$ does not adopt any product;
\ENDIF
\ENDFOR
\STATE $L := R$ // put in $L$ all nodes that adopted a product in last round
\ENDWHILE
\RETURN "Unique outcome" // No further reduction is possible
\end{scriptsize}

\end{algorithmic}
\vspace{1pt}\hrule height 0.8pt
\caption{Pseudocode for the algorithm of Theorem~\ref{thm:res3-alg}} \label{fig:alg}
\end{figure}

} 

Theorem~\ref{thm:res3} also yields an algorithm to test if a network has a unique outcome.
The algorithm simply has to simulate the contraction sequence of a network and determine whether it ends in a non-ambivalent network.
The statement of the algorithm and its analysis are omitted.

\begin{theorem} \label{thm:res3-alg}
  There exists a polynomial time algorithm, running in time $O(n^2 + n|P|)$, that determines whether a
  social network admits a unique outcome. Furthermore, if for all nodes $i$ we have $\theta(i) > \frac{1}{2}$, there is a $O(n^2)$ algorithm.
\end{theorem}

For all practical purposes we have $|P| << n$, so even for the general case the running time
would typically be $O(n^2)$.

\section{Product adoption in networks without unique outcomes}
\label{sec:adoption-analysis}

The results of the previous section reveal that many social networks will not
admit a unique outcome. In this section, we consider some natural questions
regarding product adoption
that are of interest for such networks.
We start with two optimization problems.

Suppose that a product $top$ is neither unavoidable by all nodes nor reachable.
We would like then to estimate the
worst and best-case scenario for the spread of this product. That is, starting
from a given initial network $p$, what is the minimum (resp.~maximum)
number of nodes that will adopt this product in a final network (recall
that a final network is one that has been
obtained from some initial network by a maximal sequence of reductions). Hence, the following
two problems are of interest.

\NI {\bf MIN-ADOPTION:} Given a social network $(G,P,p,\theta)$ and a product $top$, find the minimum number of nodes that adopted $top$ in a final network, starting from $(G,P,p,\theta)$.

\NI {\bf MAX-ADOPTION:} Given a social network $(G,P,p,\theta)$ and a product $top$, find the maximum number of nodes that adopted $top$ in a final network, starting from $(G,P,p,\theta)$.

We show that these two problems are substantially different, the first being essentially inapproximable while the second efficiently solvable.

\begin{theorem}
\label{thm:min-max-adoption}
If $n$ is the number of nodes of a network, then
\begin{enumerate}\smallromani
\item It is NP-hard to approximate MIN-ADOPTION with an approximation ratio better than $\Omega(n)$.
\item The MAX-ADOPTION problem can be solved in $O(n^2)$ time.
\end{enumerate}
\end{theorem}
\begin{proof}
\NI $(i)$
We give a reduction from the PARTITION problem, which is: given $n$ positive rational numbers $(a_1,\LL,a_n)$, is there a set $S$ such that $\sum_{i\in S} a_i = \sum_{i\not\in S} a_i$?
Consider an instance $I$ of PARTITION. WLOG, suppose we have normalized the numbers so that $\sum_{i=1}^n a_i = 1$. Hence the question is to decide whether there is a set $S$ such that $\sum_{i\in S} a_i = \sum_{i\not\in S} a_i = \frac{1}{2}$.

We build an instance of our problem with $3$ products, namely $P =
\{top, t, t'\}$, and with the graph shown in
Figure~\ref{fig:reduction}. The number of nodes in the line that
starts to the right of node $e$ is $M = n^{O(1)}$, hence the
reduction is of polynomial time. The weights in those edges is
$1$.  The thresholds of the nodes are $\theta(a) = \theta(b) =
\theta(c) = \theta(d) = \frac{1}{2}$, $\theta(e) = 1/2+\epsilon$, for
some $\epsilon>0$ and for the nodes to the right of $e$ we can
set the thresholds to an arbitrary positive number in $(0,1]$.
Finally, for each node $i\in\{1,\LL,n\}$, we set $w_{i,a} = w_{i, b} =
a_i$. The weights of the other edges can be seen in the figure.

\begin{figure}[htbp]
\begin{center}  \ \setlength{\epsfxsize}{5cm}
\epsfbox{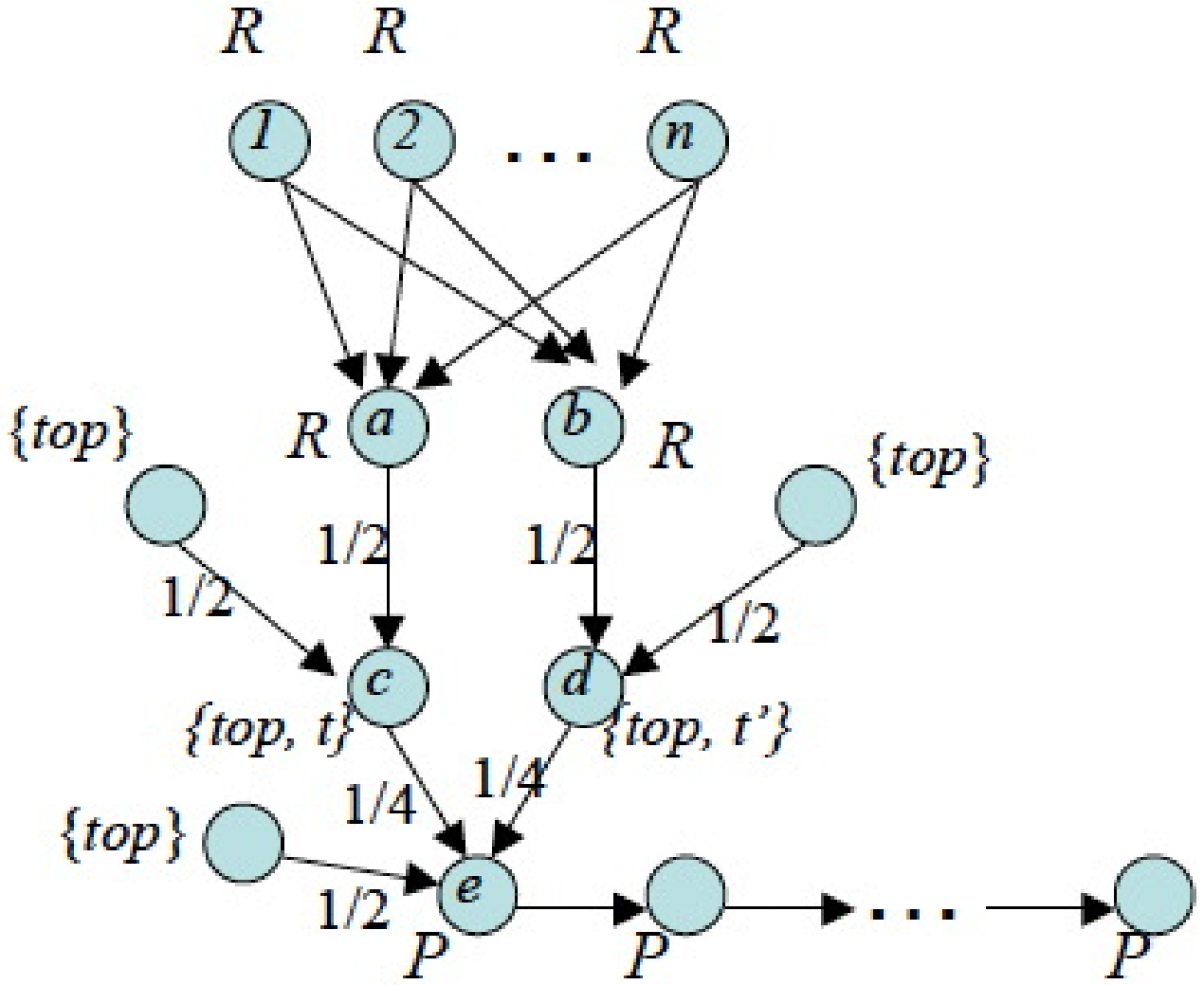}
\end{center}
\caption{The graph of the reduction with $P = \{top, t, t'\}$ and $R = \{t, t'\}$.}\label{fig:reduction}
\end{figure}

We claim that if there exists a solution to $I$, then a final
network exists where the number of nodes that adopted $top$ equals $3$, 
otherwise in all final networks the number of nodes
that adopted $top$ equals $M+5$. This directly yields the
desired result.

Suppose there is a solution $S$ to $I$. Then we can have the nodes
corresponding to the set $S$ adopt $t$ and the remaining nodes from
$\{1,\LL,n\}$ adopt $t'$. This implies that node $a$ can adopt $t$ and
node $b$ can adopt $t'$. Subsequently, node $c$ can adopt $t$ and node
$d$ can adopt $t'$, which implies that node $e$ cannot adopt any
product. Hence a final network exists in which only $3$ nodes adopted
$top$.

For the reverse direction, suppose there is no solution to the
PARTITION problem. Then, no matter how we partition the
nodes $\{1,\LL,n\}$, into $2$ sets $S, S'$, it will always be that for
one of them, say $S$, we have $\sum_{i\in S} a_i > \frac{1}{2}$,
whereas for the other we have $\sum_{i\in S'} a_i < \frac{1}{2}$.
Thus in each final network, no matter which nodes from $\{1,\LL,n\}$
adopted $t$ or $t'$, the nodes $a$ and $b$ adopted the same product.
Suppose that nodes $a$ and $b$ both adopted $t$ (the same
applies if they both adopt $t'$).  This in turn implies that
node $c$ adopted $t$ and node $d$ did not adopt $t'$. Thus, the
node $d$ could only adopt $top$.  But then the only choice for node
$e$ was to adopt $top$ and this propagates along the whole line to the
right of $e$. This completes the proof of $(i)$.

\NI $(ii)$
The algorithm for MAX-ADOPTION resembles the one used in the proof of Theorem~\ref{thm:res3-alg}.
Given the product $top$, it suffices to start with the nodes that have already adopted the product and perform fast reductions but only with respect to $top$ until no further adoption of $top$ is possible.
\end{proof}

We now move on to some decision problems that concern the behavior of a specific node
in a given social network. We consider the following natural questions.

\NI {\bf ADOPTION 1:} (unavoidable adoption of some product)  \\
Determine whether a given node has to adopt some product in all final networks.

\NI {\bf ADOPTION 2:} (unavoidable adoption of a given product) \\
Determine whether a given node has to adopt a given product in all final networks.

\NI {\bf ADOPTION 3:} (possible adoption of some product) \\
Determine whether a given node can adopt some product in some final network.

\NI {\bf ADOPTION 4:} (possible adoption of a given product) \\
Determine whether a given node can adopt a given product in some final network.

\begin{theorem}
\label{thm:adoption-problems}
The complexity of the above problems is as follows:
\begin{enumerate}\smallromani
\item ADOPTION 1 is co-NP-complete.
\item ADOPTION 2 is co-NP-complete.
\item ADOPTION 3 can be solved in $O(n^2|P|)$ time.
\item ADOPTION 4 can be solved in $O(n^2)$ time.
\end{enumerate}
\end{theorem}

The proofs of (i) and (ii) use the reduction given in the proof of Theorem \ref{thm:min-max-adoption}. We omit the proof due to lack of space.


\section{Conclusions and future work}

We have introduced a diffusion model in the presence of multiple
competing products and studied some basic questions.  We have provided
characterizations of the underlying graph structure for determining
whether a product can spread or will necessarily spread to the whole
graph, and of the networks that admit a unique outcome. We
also studied the complexity of various problems that are of interest
for networks that do not admit a unique outcome, such as the problems
of computing the minimum or maximum number of nodes that will adopt a
given product, or determining whether
a given node has to adopt some (resp.~a given) product in all final
networks.

In the proposed model, one could also incorporate game theoretic aspects by considering a strategic game either between
the nodes who decide which product to choose, or between the
producers who decide to offer their products for free to some selected
nodes. In the former case, a game theoretic analysis for players
choosing between two products has been presented in~\cite{Mor00}. An
extension with the additional option of adopting both products has
been considered in~\cite{IKMW07}.
The latter case, with the producers being the players, has been
recently studied in \cite{AFPT10} in a different model than the threshold ones.
We are particularly interested in
analyzing the set of Nash equilibria in the presence of multiple
products, as well as in introducing threshold behavior in the model of \cite{AFPT10}.

\section*{Acknowledgement}
We would like to thank Berthold V\"{o}cking for suggesting to us the first two problems in Section \ref{sec:adoption-analysis}.
{\footnotesize

\begin{thebibliography}{10}

\bibitem{AFPT10}
N.~Alon, M.~Feldman, A.~D. Procaccia, and M.~Tennenholtz.
\newblock A note on competitive diffusion through social networks.
\newblock {\em Inf. Process. Lett.}, 110(6):221--225, 2010.

\bibitem{BKS07}
S.~Bharathi, D.~Kempe, and M.~Salek.
\newblock Competitive influence maximization in social networks.
\newblock In {\em Proc. 3rd International Workshop on Internet and Network
  Economics (WINE)}, pages 306--311, 2007.

\bibitem{BFO10}
A.~Borodin, Y.~Filmus, and J.~Oren.
\newblock Threshold models for competitive influence in social networks.
\newblock In {\em Proc. 6th International Workshop on Internet and Network
  Economics (WINE 2010)}, pages 539--550, 2010.

\bibitem{CNWZ07}
T.~Carnes, C.~Nagarajan, S.~M. Wild, and A.~van Zuylen.
\newblock Maximizing influence in a competitive social network: A follower's
  perspective.
\newblock In {\em Proc. 9th International Conference on Electronic Commerce
  (ICEC)}, pages 351--360, 2007.

\bibitem{Che09}
N.~Chen.
\newblock On the approximability of influence in social networks.
\newblock {\em SIAM J. Discrete Math.}, 23(3):1400--1415, 2009.

\bibitem{EK10}
D.~Easley and J.~Kleinberg.
\newblock {\em Networks, Crowds, and Markets}.
\newblock Cambridge University Press, 2010.

\bibitem{GLM01}
J.~Goldenberg, B.~Libai, and E.~Muller.
\newblock Talk of the network: A complex systems look at the underlying process
  of word-of-mouth.
\newblock {\em Marketing Letters}, 12(3):211--223, 2001.

\bibitem{Gra78}
M.~Granovetter.
\newblock Threshold models of collective behavior.
\newblock {\em American Journal of Sociology}, 83(6):1420--1443, 1978.

\bibitem{IKMW07}
N.~Immorlica, J.~M. Kleinberg, M.~Mahdian, and T.~Wexler.
\newblock The role of compatibility in the diffusion of technologies through
  social networks.
\newblock In {\em ACM Conference on Electronic Commerce}, pages 75--83, 2007.

\bibitem{Jac08}
M.~Jackson.
\newblock {\em Social and Economic Networks}.
\newblock Princeton University Press, Princeton, 2008.

\bibitem{KKT03}
D.~Kempe, J.~M. Kleinberg, and {\'E}.~Tardos.
\newblock Maximizing the spread of influence through a social network.
\newblock In L.~Getoor, T.~E. Senator, P.~Domingos, and C.~Faloutsos, editors,
  {\em KDD}, pages 137--146. ACM, 2003.

\bibitem{KOW08}
J.~Kostka, Y.~A. Oswald, and R.~Wattenhofer.
\newblock Word of mouth: Rumor dissemination in social networks.
\newblock In {\em SIROCCO}, pages 185--196, 2008.

\bibitem{Mor00}
S.~Morris.
\newblock Contagion.
\newblock {\em The Review of Economic Studies}, 67(1):57--78, 2000.

\bibitem{MR07}
E.~Mossel and S.~Roch.
\newblock On the submodularity of influence in social networks.
\newblock In {\em STOC}, pages 128--134, 2007.

\bibitem{Sch78}
T.~Schelling.
\newblock {\em Micromotives and Macrobehavior}.
\newblock Norton, 1978.

\end{thebibliography}


}

\end{document}